\documentclass[a4paper]{article}

\usepackage{geometry}

\usepackage[pdftex,hyperindex,breaklinks]{hyperref}
\usepackage{hyperref}
\usepackage{amsmath,amsthm,nicefrac}
\usepackage{mathtools}
\usepackage{tabularx}
\usepackage{multirow}
\usepackage{thm-restate}
\usepackage{doi}

\usepackage{algorithm}
\usepackage[noend]{algpseudocode}
\algnewcommand\Input{\item[\textbf{Input:}]}
\algnewcommand\Output{\item[\textbf{Output:}]}
\algnewcommand\Required{\item[\textbf{Required:}]}
\usepackage[capitalise,noabbrev,nameinlink]{cleveref}
\Crefname{enumi}{Step}{Steps}
\Crefname{line}{Step}{Steps}

\usepackage{xspace}

\usepackage[english]{babel}
\usepackage[utf8]{inputenc}
\usepackage[bitstream-charter]{mathdesign}

\newtheorem{thm}{Theorem}[section]   
\newtheorem{lem}[thm]{Lemma}

\newcommand{\R}{\ensuremath{\mathbb{R}}\xspace}
\newcommand{\K}{\ensuremath{\mathbb{K}}}
\newcommand{\KX}{\ensuremath{\K[X]}}
\newcommand{\KXX}{\ensuremath{\K[[X]]}}
\newcommand{\M}{\mathsf{M}}

\newcommand\MP{\ensuremath{\mathsf{MP}}}
\newcommand\SP{\ensuremath{\mathsf{SP}}}
\newcommand\Inv{\ensuremath{\mathsf{Inv}}}
\newcommand\floor[1]{\left\lfloor #1\right\rfloor}
\newcommand\ceil[1]{\left\lceil #1\right\rceil}

\newcommand\cpi{c_{pi}}

\newcommand\WS{{\small WS:}\xspace}

\renewcommand\div{\operatorname{div}\,}
\DeclareMathOperator\rev{rev}
\DeclareMathOperator\loglog{loglog}

\newcommand{\bigO}[1]{\ensuremath{O(#1)}}

\makeatletter
\renewcommand\theHALG@line{\thealgorithm.\arabic{ALG@line}}
\makeatother

\title{Fast In-place Algorithms for Polynomial Operations:\\ Division, Evaluation, Interpolation}

\author{
    \hspace*{-.1\textwidth}
    \begin{minipage}{.4\textwidth}
    \centering
    Pascal Giorgi\\
    LIRMM, Univ. Montpellier, CNRS\\
    Montpellier, France\\
    \nolinkurl{pascal.giorgi@lirmm.fr}
    \end{minipage}
    \hfill
    \begin{minipage}{.4\textwidth}
    \centering
    Bruno Grenet\\
    LIRMM, Univ. Montpellier, CNRS\\
    Montpellier, France\\
    \nolinkurl{bruno.grenet@lirmm.fr}
    \end{minipage}
    \hfill
    \begin{minipage}{.4\textwidth}
    \centering
    Daniel S.\ Roche\\
    United States Naval Academy\\
    Annapolis, Maryland, U.S.A.\\
    \nolinkurl{roche@usna.edu}
    \end{minipage}
}

\begin{document}
\maketitle

\begin{abstract}
  We consider space-saving versions of several important operations on univariate polynomials, namely power series inversion and
  division, division with remainder, multi-point evaluation, and interpolation.  Now-classical results show that such problems can
  be solved in (nearly) the same asymptotic time as fast polynomial multiplication.  However, these reductions, even when applied
  to an in-place variant of fast polynomial multiplication, yield algorithms which require at least a linear amount of extra space
  for intermediate results.  We demonstrate new in-place algorithms for the aforementioned polynomial computations which require
  only constant extra space and achieve the same asymptotic running time as their out-of-place counterparts.  We also provide a
  precise complexity analysis so that all constants are made explicit, parameterized by the space usage of the underlying
  multiplication algorithms.
\end{abstract}

\section{Introduction}

Computations with dense univariate polynomials or truncated power series
over a finite ring are of central
importance in computer algebra and symbolic computation.
Since the discovery of sub-quadratic (``fast'') multiplication
algorithms~\cite{karatsuba,Cook66,SS71,harvey:hal-02070816,CK91},
a major research task was to reduce many other polynomial computations
to the cost of polynomial multiplication.

This project has been largely successful, starting with symbolic Newton
iteration for fast inversion and division with remainder~\cite{kung1974},
product tree algorithms for multi-point evaluation and interpolation
\cite{BorodinMoenck:1972:FastMultiPointEval}, the ``half-GCD'' fast Euclidean algorithm
\cite{schonhage1988}, and many more related important problems
\cite{Bostan:2003,Gathen:2013}. Not only are these problems important
in their own right, but they also form the basis for many more, such as
polynomial factorization, multivariate and/or sparse polynomial
arithmetic, structured matrix computations,
and further applications in areas such as coding theory
and public-key cryptography.

But the use of fast arithmetic frequently comes at the expense of
requiring extra \emph{temporary space} to perform the computation.
This can make a difference in practice, from the small scale where
embedded systems engineers seek to minimize hardware circuitry, to the
medium scale where a space-inefficient algorithm can exceed the
boundaries of (some level of) cache and cause expensive cache misses, to
the large scale where main memory may simply not be sufficient to hold
the intermediate values.
In a streaming model, where the output must be written only once, in
order, explicit time-space tradeoffs prove that fast multiplication
algorithms will always require up to linear extra space. And indeed, all
sub-quadratic polynomial multiplication algorithms we are aware of ---
in their original formulation --- require linear extra space
\cite{karatsuba,Cook66,SS71,harvey:hal-02070816,CK91}.

However, if we
treat the output space as pre-allocated random-access memory, allowing
values in output registers to be both read and written multiple times,
then improvements are possible. In-place quadratic-time algorithms for
polynomial arithmetic are described in~\cite{Monagan93}.
A series of recent results provide
explicit algorithms and reductions from arbitrary fast multiplication
routines which have the same \emph{asymptotic} running time, but use
only constant extra space~\cite{roche:2009,hr10, GioGreRo19}.
That is, these algorithms trade a \emph{constant} increase in the
running time for a \emph{linear} reduction in the amount of extra space.
So far, these results are limited to multiplication routines and related
computations such as middle and short product.
Applying in-place multiplication algorithms directly to other problems,
such as those considered in this paper, does not immediately yield an
in-place algorithm for the desired application problem.

\subsection{Our work}

\begin{table*}[htbp]
\hspace*{-.1285\textwidth}
\begin{tabularx}{1.257\textwidth}{|c|c|c|X|}
  \cline{2-4}
\multicolumn{1}{c|}{}&\textbf{Time}  & \textbf{Space} & \textbf{Reference} \\\hline 
  \multirow{1}{*}{\textbf{Power series inversion} }&$(\lambda_m+\lambda_s)M(n)$ & $\frac{1}{2}\max(c_m, c_s+1)n$ & \cite[Alg.~\texttt{MP-inv}]{HaQueZi04} \\
  \multirow{1}{*}{at precision $n$}                &$\lambda_m\M(n)\log_{\frac{c_m+2}{c_m+1}}(n)$ & $\bigO 1$ & \Cref{thm:inplace-inversion} \\ \hline 

&$(\lambda_m+\frac{3}{2}\lambda_s)\M(n)$ & $\frac{c_m+1}{2} n$ & \cite[Alg.~\texttt{MP-div-KM}]{HaQueZi04}\\
\multirow{1}{*}{\textbf{Power series division}}&$\lambda_m\M(n)\log_{\frac{c_m+3}{c_m+2}}(n)$&$\bigO{1}$ & \Cref{thm:inplace-division}\\
\multirow{1}{*}{at precision $n$}              &$\bigO{\M(n)}$ & $\alpha n$, for any $\alpha > 0$ & Remark~\ref{rem:division-linear}\\
                                               &$\left(\lambda_m(\frac{c+1}{2}+\frac{1}{c})+\lambda_s(1+\frac{1}{c})\right)\M(n)$ &
                                               $\bigO{1}^\ddagger$ & \cref{cor:division-inplace-erasing}\\
\hline 
\multirow{1}{*}{\textbf{Euclidean division}}   & $(\lambda_m+\frac{3}{2}\lambda_s)\M(m)+\lambda_s\M(n)$ & $\max(\frac{c_m+1}{2} m-n, c_s n)$ & standard algorithm\\ 
\multirow{1}{*}{\textbf{of polynomials}}
  & $ 2\lambda_s\M(m)+(\lambda_m+\lambda_s)\M(n) $ & $(1+\max(\frac{c_m}{2}, \frac{c_s+1}{2},c_s))n$ & $\ceil{\frac{m}{n}}$ balanced div. (precomp)\\
\multirow{1}{*}{in sizes ($m+n-1,n$)}
     & $\left(\lambda_m(\frac{c+1}{2}+\frac{1}{c})+ \lambda_s(2+\frac{1}{c})\right)\M(m)$ & $\bigO{1}$ &Theorem~\ref{thm:complexity-Euclidean-division}\\\hline
  \multirow{1}{*}{\textbf{multipoint evaluation} }   & $\nicefrac{3}{2}\M(n)\log(n)$ &  $n \log(n)$ & \cite{Bostan:2003} \\
  \multirow{1}{*}{size-$n$ polynomial on $n$ points} &
    $\nicefrac{7}{2}\M(n)\log(n)$    & $n$ & \cite{vonzurGathen1992}, \cref{lem:mpeval-linspace-balanced} \\
  \multirow{1}{*}{}                                  &
    \((4+2\lambda_s/\log(\frac{c_s+3}{c_s+2}))\M(n)\log(n)\)
    & $\bigO{1}$ &  \cref{thm:inplace-multieval}\\
  \hline
  \multirow{1}{*}{\textbf{interpolation} }
    & \(\nicefrac{5}{2}\M(n)\log(n)\) & \(n\log(n)\) & \cite{Bostan:2003} \\
  \multirow{1}{*}{size-$n$ polynomial on $n$ points}
    & \(5\M(n)\log(n)\) & \(2n\) & \cite{Gathen:2013,vonzurGathen1992}, \cref{lem:interp-linspace} \\
  \multirow{1}{*}{}                                  & $\simeq 105\M(n)\log(n)$& $\bigO{1}$& Theorem~\ref{thm:interpol} \\
  \hline
\end{tabularx}
\caption{Summary of complexity analyses,
  \mdseries
  omitting non-dominant terms and assuming \(c_f\le{}c_s\le{}c_m\).
  We use $c =c_m+3$. For $\bigO{1}^\ddagger$ space, the memory model is changed such that the input dividend can be overwritten. Here, and throughout the paper, the base of the logarithms is $2$ if not otherwise stated.}
\label{table:comp}
\end{table*}

In this paper, we present new in-place algorithms for power series inversion and division, polynomial division with remainder,
multi-point evaluation, and interpolation. These algorithms are \emph{fast} because their running time is only a constant time
larger than the fastest known out-of-place algorithms, parameterized by the cost of dense polynomial multiplication.

Our space complexity model is the one of~\cite{roche:2009,hr10,GioGreRo19} where input space is read only while output space is
pre-allocated and can be used to store intermediate results. In that model, the space complexity is measured by only counting the
auxiliary space required during the computation, excluding input and output spaces. We shall mention that a single memory location
or register may contain either an element of the coefficient ring, or a pointer to the input or output space.  It follows that
in-place algorithms are those that require only a constant number of extra memory locations.

For all five problems, we present in-place variants which have nearly the same asymptotic running time as their fastest
out-of-place counterparts. The power series inversion and division algorithms incur an extra \(\log(n)\) overhead when
quasi-linear multiplication is used, while the polynomial division, evaluation, and interpolation algorithms keep the same
asymptotic runtime as the fastest known algorithm.  Our reductions essentially trade a small amount of extra runtime for a
significant decrease in space usage.

Our motivation in this work is mainly theoretical. We address the existence of such fast in-place algorithms as we already did for
polynomial multiplications~\cite{GioGreRo19}. To further extend our result, we compare precisely the number of arithmetic
operations in our algorithms with the best known theoretical bounds. These results are summarized in \cref{table:comp}.

Of course further work is needed to determine the practicability of our approach. In particular cache misses play a predominant
role when dealing with memory management. Studying the cache complexity of all these algorithms, for instance in the idealized
cache model~\cite{Frigo99}, would give more precise insights. However, the practicability will heavily depend on the underlying
multiplication algorithms.  Due to their diversity and the need for fine-tuned implementations, we leave this task to future work.

\subsection{Notation}
By a size-$n$ polynomial, we mean a polynomial of degree $\le n-1$.  As usual, we denote by $\M(n)$ a bound on the number of
operations in $\K$ to multiply two size-$n$ polynomials, and we assume that \(\alpha\M(n)\le\M(\alpha{}n)\) for any constant
\(\alpha\ge{}1\).  All known multiplication algorithms have at most a linear space complexity. Nevertheless, several results
reduce this space complexity at the expense of a slight increase in the time complexity
\cite{Thome:2002,roche:2009,hr10,GioGreRo19}. To provide tight analyses, we consider multiplication algorithms with time
complexity $\lambda_f \M(n)$ and space complexity $c_fn$ for some constants $\lambda_f\ge 1$ and $c_f\ge 0$.

Let us recall that the middle product of a size-$(m+n-1)$ polynomial $F\in\KX$ and a size-$n$ polynomial $G\in\KX$ is the size-$m$
polynomial defined as $\MP(F,G)= (FG \div X^{n-1}) \bmod X^m$.  We denote by $\lambda_m \M(n)$ and $c_m n$ the time and space
complexities of the middle product of size $(2n-1,n)$.  Then, a middle product in size $(m + n - 1,n)$ where $m<n$ can be computed
with $\ceil{\frac{n}{m}} \lambda_m\M(m)$ operations in $\K$ and $(c_m+1) m$ extra space.  Similarly, the short product of two
size-$n$ polynomials $F,G\in\KX$ is defined as $\SP(F,G)=FG \bmod X^n$ and we denote by $\lambda_s\M(n)$ and $c_s n$ its time and
space complexities.

On the one hand, the most time-efficient algorithms achieve $\lambda_f = \lambda_m = \lambda_s = 1$ while $2\leq c_f$, $c_m$,
$c_s \leq 4$, using the \emph{Transposition principle}~\cite{HaQueZi04,Bostan:2003} for $\lambda_m=\lambda_f$. On the other hand,
the authors recently proposed new space-efficient algorithms reaching $c_f = 0$, $c_m = 1$ and $c_s = 0$ while $\lambda_f$,
$\lambda_m$ and $\lambda_s$ remain constants~\cite{GioGreRo19}.

Writing $F=\sum_{i=0}^{d}f_iX^i\in\KX$, we will use $\rev(F)\in\KX$ to denote the reverse polynomial of $F$, that is,
$\rev(F)=X^{d}F(1/X)$, whose computation does not involve any operations in $\K$.  Note that we will use abusively the notation
$F_{[a..b[}$ to refer to the chunk of $F$ that is the polynomial $\sum_{i=a}^{b-1}f_iX^i$, and the notation $F_{[a]}$ for the
coefficient $f_a$.  Considering our storage, the notation $F_{[a..b[}$ will also serve to refer to some specific registers
associated to $F$.  When necessary, our algorithms indicate with {\small WS} the output registers used as work space.

\section{Inversion and divisions}
\label{sec:invdiv}

In this section, we present in-place algorithms for the inversion and the division of power series as well as the Euclidean
division of polynomials. As a first step, we investigate the space complexity from the literature for these
computations. 

\subsection{Space complexity of classical algorithms} \label{ssec:classical-complexity}

\paragraph{Power series inversion}
Power series inversion is usually computed through Newton iteration: If $G$ is the inverse of $F$ at precision $k$ then
$H = G + (1-GF)G \mod X^{2k}$ is the inverse of $F$ at precision $2k$. This allows one to compute $F^{-1}$ at precision $n$ using
$\bigO{\M(n)}$ operations in $\K$, see~\cite[Chapter~9]{Gathen:2013}.  As noticed in~\cite[Alg. \texttt{MP-inv}]{HaQueZi04} only
the coefficients of degree $k$ to $2k-1$ of $H$ are needed. Thus, assuming that $G_{[0..k[}=F^{-1}\bmod X^k$, one step of Newton
iteration computes $k$ new coefficients of $F^{-1}$ into $G_{[k..2k]}$ as 
\begin{equation}\label{eq:std-inv}
G_{[k..2k[} = -\SP(\MP( F_{[1..2k[},G_{[0..k[}), G_{[0..k[}).
\end{equation}
The time complexity is then $(\lambda_m+\lambda_s)\M(n)$ for an inversion at precision $n$. For space complexity, the most
consuming part is the last iteration of size $\frac{n}{2}$. It needs $\max(c_m,c_s+1)\frac{n}{2}$ extra registers: One can compute
the middle product in $G_{[\frac{n}{2}..n[}$ using $c_m\frac{n}{2}$ extra registers, then move it to $\frac{n}{2}$ extra registers
and compute the short product using $c_s\frac{n}{2}$ registers.

\paragraph{Power series division}
Let $F,G\in\KXX$, the fast approach to compute $F/G \bmod X^n$ is to first invert $G$ at precision $n$ and then to multiply the
result by $F$. The complexity is given by one inversion and one short product at precision $n$. Actually, Karp and Markstein
remarked in~\cite{KaMa97} that $F/G$ can be directly computed during the last iteration. Applying this trick, the complexity
becomes $(\lambda_m+\frac{3}{2}\lambda_s)\M(n)$~\cite{HaQueZi04}, see also~\cite{Bernstein:Fastmul:2008}.  The main difference
with inversion is the storage of the short product of size $\frac{n}{2}$, yielding a space complexity of
$\max(c_m+1,c_s+1)\frac{n}{2}$.

\paragraph{Euclidean division of polynomials}
Given two polynomials $A,B $ of respective size $m+n-1$ and $n$, the fast Euclidean division computes the quotient $A \div B$ as
$\rev(\rev(A)/\rev(B))$ viewed as power series at precision $m$~\cite[Chapter~9]{Gathen:2013}. The remainder $R$ is retrieved with
a size-$n$ short product, yielding a total time complexity of $(\lambda_m+\frac{3}{2}\lambda_s)\M(m) + \lambda_s\M(n)$.  Since the
remainder size is not determined by the input size we assume that we are given a maximal output space of size $n-1$.  As this
space remains free when computing the quotient, this step requires $\tfrac{1}{2}\max(c_m+1,c_s+1)m-n+1$ extra space, while
computing the remainder needs $c_sn$.

As a first result, when $m\le n$, using space-efficient multiplication is enough to obtain an in-place $\bigO{\M(n)}$ Euclidean
division. Indeed, the output space is enough to compute the \emph{small} quotient, while the remainder can be computed in-place
\cite{GioGreRo19}.

When $m>n$, the space complexity becomes $\bigO{m-n}$.  In that case, the Euclidean division of $A$ by $B$ can also be computed by
$\ceil{\frac{m}{n}}$ \emph{balanced} Euclidean divisions of polynomials of size $2n-1$ by $B$. It actually corresponds to a
variation of the \emph{long division algorithm}, in which each step computes $n$ new coefficients of the quotient.  To save some
time, one can precompute the inverse of $\rev(B)$ at precision $n$, which gives a time complexity
$(\lambda_m+\lambda_s)\M(n) + \frac{m}{n}2\lambda_s\M(n)\leq 2\lambda_s\M(m)+(\lambda_m+\lambda_s)\M(n) $ and space complexity
$(1+\max(\frac{c_m}{2}, \frac{c_s+1}{2}, c_s))n$.

Finally, one may consider to only compute the quotient or the remainder. Computing quotient only is equivalent to power series
division. For the computation of the remainder, it is not yet known how to compute it without the quotient. In that case, we shall
consider space usage for the computation and the storage of the quotient.  When $m$ is large compared to $n$, one may notice that
relying on balanced divisions does not require one to retain the whole quotient, but only its $n$ latest computed coefficients. In
that case the space complexity only increases by $n$.  Since we can always perform a middle product via two short products, we
obtain the following result.

\begin{lem}\label{lem:remainder}
  Given $A\in\KX$ of size $m$ and $B\in\KX$, monic of size $n$,
  and provided $n$ registers for the output,
  the remainder $A\bmod{}B$ can be computed using
  \(2\lambda_s\M(m) + 3\lambda_s\M(n) + \bigO{m+n}\)
  operations in $\K$ and
  \((c_s+2)n\)
  extra registers.
\end{lem}

\subsection{In-place power series inversion}

We notice that during the first Newton iterations, only a few coefficients of the inverse have been already written. The output
space thus contains lots of free registers, and the standard algorithm can use them as working space. In the last iterations, the
number of free registers becomes too small to perform a standard iteration. Our idea is then to \emph{slow down} the computation.
Instead of still doubling the number of coefficients computed at each iteration, the algorithm computes less and less
coefficients, in order to be able to use the free output space as working space. We denote these two phases as acceleration and
deceleration phases.

The following easy lemma generalizes Newton iteration to compute only $\ell\le k $ new coefficients from an inverse at precision $k$.
\begin{lem}\label{lem:inplace-inv}
  Let $F$ be a power series and $G_{[0..k[}$ contain its inverse at precision $k$.  Then for $0 <\ell\le k$, if we compute
\begin{equation}\label{eq:inplace-inv}
  G_{[k..k+\ell[} = -\SP\left(\MP\left(F_{[1..k+\ell[},G_{[0..k[}\right), G_{[0..\ell[}\right)
\end{equation}
then $G_{[0..k+\ell[}$ contains the inverse of $F$ at precision $k+\ell$.
\end{lem}

Algorithm~\ref{alg:inplace-inversion} is an in-place fast inversion algorithm.  Accelerating and decelerating phases correspond to
$\ell = k$ and $\ell < k$.

\begin{algorithm}
\caption{In-Place Fast Power Series Inversion (\textsc{InPlaceInv})}
\begin{algorithmic}[1]
  \Input $F\in\KX$ of size $n$, such that $F_{[0]}$ is invertible;
  \Output $G\in\KX$ of size $n$, such that $FG = 1\mod X^n$.
  \Required $\MP$ and $\SP$ alg. using extra space  $\le c_mn$ and $\le c_sn$.
  \State $G_{[0]}\gets F_{[0]}^{-1}$
  \State $k\gets 1$,\quad $\ell\gets{}1$
  \While{$\ell > 0$}
    \State $G_{[n-\ell..n[}\gets \MP(F_{[1..k+\ell[},G_{[0..k[})$ \Comment \WS $G_{[k..n-\ell[\phantom{+\ell}}$
      \label{step:inv:mp1}
    \State $G_{[k..k+\ell[}\gets \SP(G_{[0..\ell[}, -G_{[n-\ell..n[})$\Comment  \WS $G_{[k+\ell..n-\ell[}$
      \label{step:inv:sp1}
    \State $k \gets k+\ell$
    \State $\ell\gets \min\left(k, \floor{\frac{n-k}{c}}\right)$ where $c = 2+\max(c_m,c_s)$
  \EndWhile
\State $G_{[k..n[} \gets \SP(G_{[0..n-k[},-\MP(F_{[1..n[},G_{[0..k[}))$ \Comment $\bigO{1}$ space
    \label{step:inv:final}
\end{algorithmic}
\label{alg:inplace-inversion}
\end{algorithm}

\begin{restatable}{thm}{inplaceinversion} \label{thm:inplace-inversion}
  Algorithm \ref{alg:inplace-inversion} is correct.
  It uses $O(1)$ space, and either
  $\lambda_m\M(n) \log_{\frac{c_m+2}{c_m+1}}(n) + \bigO{\M(n)}$ operations in $\K$
  when $\M(n)$ is quasi-linear, or
  $\bigO{\M(n)}$ operations in $\K$ when $\M(n)=n^{1+\gamma}$, $0<\gamma\le 1$.
\end{restatable}

\begin{proof}
  \Cref{step:inv:mp1,step:inv:sp1}, and \cref{step:inv:final}, correspond to Equation~\eqref{eq:inplace-inv}. They compute $\ell$ new coefficients of $G$ when $k$ of them
  are already written in the output, whence
Lemma~\ref{lem:inplace-inv} implies the
  correctness.

  \Cref{step:inv:mp1} needs $(c_m+2)\ell$ free registers for its computation and its storage. Then $(c_s+2)\ell$ free registers
  are needed to compute $\SP(G_{[0..\ell[}, G_{[n-\ell..n[})$ using $\ell$ registers for $G_{[n-\ell..n[}$ and $(c_s+1)\ell$
  registers for the short product computation and its result.  For this computation to be done in-place, we need $c\ell\le
  n-k$. Since at most $k$ new coefficients can be computed, the maximal number of new coefficients in each step is
  $ \ell = \min\left(k, \floor{\frac{n-k}{c}}\right).  $
    
  Each iteration uses $\bigO{\M(k)}$ operations in $\K$: $\bigO{\ceil{k/\ell}\M(\ell)}$ for the middle product at
  \cref{step:inv:mp1} and $\bigO{\M(\ell)}$ for the short product at \cref{step:inv:sp1}.  The accelerating phase stops when
  $k > \frac{n-k}{c+1}$, that is, $k>\frac{n}{c+2}$.  It costs $\sum_{i=0}^{\floor{\log\frac{n}{c+2}}} \M(2^i) = \bigO{\M(n)}$.
  During the decelerating phase, each iteration computes a constant fraction of the remaining coefficients. Hence, this phase
  lasts for $\delta=\log_{\frac{c}{c-1}} n$ steps.

  Let $\ell_i$ and $k_i$ denote the values of $\ell$ and $k$ at the $i$-th iteration of the deceleration phase and $t_i=n-k_i$.
  Then one iteration of the deceleration phase costs one middle product in sizes $( n-t_i+\floor{\frac{t_i}{c}}-1,n-t_i)$ and one
  short product in size $\floor{\frac{t_i}{c}}$. The total cost of all the short products amounts to
  $\sum_{i}\M(t_i)= \bigO{\M(n)}$ since $\sum_{i}t_i \leq cn$. The cost of the middle product at the $i$-th step is
  \[
    \lambda_m\ceil{{(n-t_i)}/{\floor{\tfrac{t_i}{c}}}} \M\left(\floor{\tfrac{t_i}{c}}\right) = \lambda_m \M(n) +\bigO{n}.
  \]
  Therefore, the total cost of all the middle products is at most $\lambda_m \M(n)\log_{\frac{c}{c-1}}(n) + \bigO{\M(n)}$ and is
  dominant in the complexity.  We can choose the in-place short products of~\cite{GioGreRo19} and get $c=c_m+2$.  The complexity
  is then $\lambda_m\M(n) \log_{\frac{c_m+2}{c_m+1}}(n) + \bigO{\M(n)}$.

  If $\M(n)=n^{1+\gamma}$ with $0<\gamma\le 1$, the cost of each iteration is
  $\bigO{\ceil{\frac{n-t_i}{\ell_i}}\ell_i^{1+\gamma}}$.  Since $\ell_0\le n$, we have $\ell_i < n (\frac{c-1}{c})^i +c$, whence
  \[
    \sum_{i=1}^\delta\ceil{\frac{n-t_i}{\ell_i}}\ell_i^{1+\gamma} \le n  \sum_{i=1}^\delta\ell_i^\gamma \le  n \sum_{i=1}^\delta \left(n \left(\frac{c-1}{c}\right)^i +c\right)^\gamma.
  \]
  Since $0<\gamma\leq 1$, we have $(\alpha+\beta)^\gamma\leq \alpha^\gamma+\beta^\gamma$ for any $\alpha,\beta>0$, and the
  complexity is $ n^{1+\gamma}\sum_{i=1}^\delta \left(\frac{c-1}{c}\right)^{i\gamma} +\bigO{n\log n} = \bigO{\M(n)}.  $
\end{proof}

\subsection{In-place division of power series}\label{ssec:serie-div}

Division of power series can be implemented easily as an inversion followed by a product. Yet, using in-place algorithms for these
two steps is not enough to obtain an in-place division algorithm since the intermediate result must be stored. Karp and
Markstein's trick, that includes the dividend in the last iteration of Newton iteration~\cite{KaMa97}, cannot be used directly in
our case since we replace the very last iteration by several ones. We thus need to build our in-place algorithm on the following
generalization of their method.

\begin{lem}\label{lem:psdiv}
Let $F$ and $G$ be two power series, $G$ invertible, and $Q_{[0..k[}$ contain their quotient at precision $k$. Then for $0 < \ell \le k$, if we compute
\[ Q_{[k..k+\ell[} =  \SP\left(G^{-1}_{[0..\ell[}, F_{[k..k+\ell[} - \MP(G_{[1..k+\ell[}, Q_{[0..k[})\right)\]
then $Q_{[0..k+\ell[}$ contains their quotient at precision $k+\ell$.
\end{lem}

\begin{proof}
  Let us write $F/G = Q_k + X^k Q_\ell + \bigO{X^{k+\ell}}$. We prove that $Q_\ell = G^{-1}\times((F-GQ_k)\div X^k)\bmod
  X^\ell$. By definition, $F \equiv G(Q_k+X^kQ_\ell)\bmod X^{k+\ell}$. Hence $(F-GQ_k)\div X^k = GQ_\ell\bmod X^\ell$. Therefore,
  $Q_\ell = (G^{-1} \times ((F-GQ_k)\div X^k))\bmod X^\ell$.  Finally, since only the coefficients of degree $k$ to $k+\ell-1$ of
  $GQ_k$ are needed, they can be computed as $\MP(G_{[1..k+\ell[},Q_{[0..k[})$.
\end{proof}

Algorithm~\ref{alg:psdiv} is an in-place power series division algorithm based on Lemma~\ref{lem:psdiv}, choosing at each step the
appropriate value of $\ell$ so that all computations can be performed in place.

\begin{algorithm}
\caption{In-Place Power Series Division (\textsc{InPlacePSDiv})}
\begin{algorithmic}[1]
  \Input $F, G\in\KX$ of size $n$, such that $G_{[0]}$ is invertible;
  \Output $Q\in\KX$ of size $n$, such that $F/G = Q\mod X^n$.
  \Required $\MP$, $\SP$, $\Inv$ alg. using extra space $\leq c_mn,~c_sn,~c_i n$.
  \State $k\gets \floor{n/\max(c_i+1,c_s+2)}$ \label{psdiv:t}
  \State $Q_{[n-k..n[}\gets \rev(\Inv(G_{[0..k[}))$ \Comment \WS $Q_{[0..n-k[}$ \label{psdiv:inv}
  \State $Q_{[0..k[}\gets \SP(F_{[0..k[}, \rev(Q_{[n-k..n[}))$ \Comment \WS $Q_{[k..n-k[}$ \label{psdiv:sp}
  \State $\ell\gets \floor{(n-k)/(3+\max(c_m,c_s))}$
  \While{$\ell > 0$}
    \State $Q_{[n-2\ell..n-\ell[}\gets \MP(G_{[1..k+\ell[}, Q_{[0..k[})$ \Comment \WS $Q_{[k..n-2\ell[}$\label{psdiv:mp}
    \State $Q_{[n-2\ell..n-\ell[}\gets F_{[k..k+\ell[} - Q_{[n-2\ell..n-\ell[}$\label{psdiv:diff}
    \State  \emph{let us define $Q_\ell^*=\rev(Q_{[n-\ell..n[})$} \label{psdiv:splo}
    \Statex \hspace{.85cm}$Q_{[k..k+\ell[}\gets \SP(Q_{[n-2\ell..n-\ell[},Q_\ell^*)$ \Comment \WS $Q_{[k+\ell..n-2\ell[}$
    \State $k\gets k+\ell$
    \State $\ell\gets \floor{(n-k)/(3+\max(c_m,c_s))}$ \label{psdiv:s}
  \EndWhile
  \State $tmp\gets F_{[k..n[}- \MP(G_{[1..n[},Q_{[0..k[})$ \Comment constant space
    \label{psdiv:mp2}
  \State $Q_{[k..n[}\gets \SP(tmp,\rev(Q_{[k..n[}))$ \Comment constant space
    \label{psdiv:sp2}
\end{algorithmic}
\label{alg:psdiv}
\end{algorithm}

\begin{restatable}{thm}{inplacedivision}\label{thm:inplace-division}
  Algorithm \ref{alg:psdiv} is correct. It uses $\bigO{1}$ space, and either
  $\lambda_m\M(n) \log_{\frac{c_m+3}{c_m+2}}(n) + \bigO{\M(n)}$ operations in $\K$ when $\M(n)$ is quasi-linear or $\bigO{\M(n)}$
  operations in $\K$ when $\M(n)=\bigO{n^{1+\gamma}}$, $0<\gamma \le 1$.
\end{restatable}

\begin{proof}
  The correctness follows from Lemma~\ref{lem:psdiv}.  The inverse of $G$ is computed once at \cref{psdiv:inv}, at precision
  $\floor{n/\max(c_i+1,c_s+2)}$. Its coefficients are then progressively overwritten during the loop since \cref{psdiv:splo} only
  requires $\ell$ coefficients of the inverse, and $\ell$ is decreasing. Since $c_i = \frac{1}{2}\max(c_m,c_s+1)$, $\ell$ is
  always less than the initial precision.  For simplicity of the presentation, we store the inverse in reversed order in
  $Q_{[n-k..n[}$.  \Cref{psdiv:inv} requires space $c_ik$ while the free space has size $n-k$: Since $k\le \frac{n}{c_i+1}$, the
  free space is large enough. Similarly, the next step requires space $c_sk$ while the free space has size $n-2k$, and
  $k\le\frac{n}{c_s+2}$. \Cref{psdiv:mp} needs $(c_m+1)\ell$ space and the free space has size $n-k-2\ell$, and \cref{psdiv:splo}
  requires $c_s\ell$ space while the free space has size $n-k-3\ell$. Since $\ell \le \frac{n-k}{3 + \max(c_m,c_s)}$, these
  computations can also be performed in place.

  The time complexity analysis is very similar to the one of \cref{alg:inplace-inversion} given in \cref{thm:inplace-inversion}. The main
  difference is \cref{psdiv:diff} which adds a negligible term $\bigO{n\log n}$ in the complexity.
\end{proof}

\begin{restatable}{cor}{divisioninplaceerasing}\label{cor:division-inplace-erasing}
  If it can erase its dividend, Algorithm~\ref{alg:psdiv} can be modified to improve its complexity to
  $ \left(\lambda_m(\frac{c+1}{2}+\frac{1}{c})+ \lambda_s(1+\frac{1}{c})\right)\M(n) + \bigO{n}$ operations in $\K$ where
  $c=\max(c_m+3,c_s+2)$, still using $\bigO{1}$ extra space.
\end{restatable}

\begin{proof}
  Once $k$ coefficients of $Q$ have been computed, $F_{[0..k[}$ is not needed anymore. This means that at \cref{psdiv:diff}, the result can
  be directly written in $F_{[k..k+\ell[}$ and that $F_{[0..k[}$ can be used as working space in the other steps of the loop. The free space
  at \cref{psdiv:mp,psdiv:splo} becomes $n-2\ell$ instead of $n-k-2\ell$ and $n-k-3\ell$ respectively. Therefore, $\ell$
  can always be chosen as large as $\floor{\frac{n}{c}}$ where $c = \max(c_m+3,c_s+2)$. Since $\ell$ stays positive, we also modify the
  algorithm to stop when all the coefficients of $Q$ have been computed.
  
  To simplify the complexity analysis, we further assume that $k$ gets the same value $\floor{\frac{n}{c}}$ at \cref{psdiv:t}.
  \Cref{psdiv:inv} requires $(\lambda_s+\lambda_m)\M(\floor{\frac{n}{c}})$ operations in $\K$. 
  The sum of the input sizes of all the short products in the algorithm is $n$. Their total complexity is thus $\lambda_s\M(n)$.  
  At the $i$-th iteration of the loop, $k = (i+1)\ell$. Therefore \cref{psdiv:mp} has complexity $i\floor{\frac{n}{c}}$. 
  \Cref{psdiv:diff}  requires $\floor{\frac{n}{c}}$ operations in $\K$.
  Altogether, the complexity of the modified algorithm is 
  \[ \lambda_s\M(n) + (\lambda_s+\lambda_m) \M\left(\floor{\tfrac{n}{c}}\right)
        + \sum_{i=1}^c \left( i\lambda_m \M\left(\floor{\tfrac{n}{c}}\right) +\floor{\tfrac{n}{c}}\right) \]
  which is $\left(\lambda_m(\frac{c+1}{2}+\frac{1}{c})+ \lambda_s(1+\frac{1}{c})\right)\M(n) + \bigO{n}$.
\end{proof}

Using similar techniques, we get the following variant. 

\begin{restatable}{rem}{divisionlinear}\label{rem:division-linear}
  Algorithm~\ref{alg:psdiv} can be easily modified to improve the complexity to $\bigO{M(n)}$ operations in $\K$ when a linear
  amount of extra space is available, say $\alpha n$ registers for some $\alpha \in \R_+$.
\end{restatable}

\subsection{In-place Euclidean division of polynomials}\label{ssec:poly_division}

If $A$ is a size-$(m+n-1)$ polynomial and $B$ a size-$n$ polynomial, one can compute their size-$m$ quotient $Q$ in place using
Algorithm~\ref{alg:psdiv}, in $O((\M(m)\log m))$ operations in $\K$. When $Q$ is known, the remainder $R = A - BQ$, can be
computed in-place using $O(\M(n))$ operations in $\K$ as it requires a single short product and some subtractions.  As already
mentioned, the exact size of the remainder is not determined by the size of the inputs.  Given any tighter bound $r<n$ on
$\deg(R)$, the same algorithm can compute $R$ in place, in time $\bigO{\M(r)}$.

Altogether, we get in-place algorithms that compute the quotient of two polynomials in time $O(\M(m)\log m)$, or the quotient and
size-$r$ remainder in time $O(\M(m)\log m + \M(r))$.  As suggested in Section~\ref{ssec:classical-complexity} and in
Remark~\ref{rem:division-linear}, this complexity becomes $O(\M(m)+ \M(r))$ whenever $m=\bigO{r}$.  Indeed, in that case the remainder
space can be used to speed-up the quotient computation.  We shall mention that computing only the remainder remains a harder
problem as we cannot count on the space of the quotient while it is required for the computation. As of today, only the
classical quadratic long division algorithm allows such an in-place computation.

We now provide a new in-place algorithm for computing both the quotient and the remainder that achieves a complexity of $O(\M(m)+\M(n))$
operation in $\K$ when $m\ge n$.
Our algorithm requires an output space of size $n-1$ for the remainder since taking any smaller size $r<n-1$ would rebind to power series division.

\begin{algorithm}
\caption{In-Place Euclidean Division (\textsc{InPlaceEuclDiv})}
\begin{algorithmic}[1]
  \Input $A, B\in\KX$ of sizes $(m + n,n)$, $m\ge n$, such that $B_{[0]}\neq 0$;
  \Output $Q, R\in\KX$ of sizes $(m+1,n-1)$ such that $A = BQ+R$;
  \Required  In-place $\textsc{DivErase}(F,G,n)$ computing $F/G \bmod X^n$ while erasing $F$; In-place \SP; 
  \Statex \emph{For simplicity, $H$ is a size-n polynomial  such that $H_{[0..n-1[}$ is $R$ and $H_{[n-1]}$ is an extra register}
  \State $H\gets A_{[m..m+n[}$
  \State $k\gets m+1$ 
  \While{$k > n$}
    \State $Q_{[k-n..k[}\gets \rev(\textsc{DivErase}(\rev(H),\rev(B),n))$ \label{eucldiv:diverase}
    \State $H_{[0..n-1[}\gets\SP(Q_{[k-n..k-1[}, B_{[0..n-1[})$ \label{eucldiv:sp}
    \State $H_{[1..n[}\gets A_{[k-n..k-1[}-H_{[0..n-1[}$
    \State $H_{[0]}\gets A_{[k-n-1]}$
    \State $k\gets k-n$
  \EndWhile
  \State $Q_{[0..k[}\gets \rev(\textsc{DivErase}(\rev(H_{[n-k..n[}),  \rev(B_{[n-k..n[})))$
  \State $H_{[0..n-1[}\gets\SP(Q_{[0..n-1[}, B_{[0..n-1[})$
  \State $H_{[0..n-1[}\gets A_{[0..n-1[} - H_{[0..n-1[}$
  \State \Return $(Q,H_{[0..n-1[})$
\end{algorithmic}
\label{alg:eucldiv}
\end{algorithm}

\begin{thm}\label{thm:complexity-Euclidean-division}
  Algorithm~\ref{alg:eucldiv} is correct. It uses $\bigO{1}$ extra space and $\left(\lambda_m(\frac{c+1}{2}+\frac{1}{c})+ \lambda_s(2+\frac{1}{c})\right)\M(m)+ \bigO{m \log n}$ operations in $\K$ where $c = \max(c_m+3,c_s+2)$.
\end{thm}

\begin{proof}
  Algorithm~\ref{alg:eucldiv} is an adaptation of the classical \emph{long division algorithm}, recalled in
  Section~\ref{ssec:classical-complexity}, where chunks of the quotient are computed iteratively \emph{via} Euclidean division of
  size $(2n-1,n)$. The main difficulty is that the update of the dividend cannot be done on the input. Since we compute only
  chunks of size $n$ from the quotient, the update of the dividend affects only $n-1$ coefficients. Therefore, it is possible to
  use the space of $R$ for storing these new coefficients. As we need to consider $n$ coefficients from the dividend to get a new
  chunk, we add the missing coefficient from $A$ and consider the polynomial $H$ as our new dividend.

  By \cref{cor:division-inplace-erasing}, \cref{eucldiv:diverase} can be done in place while erasing $H$, which is not part of the
  original input. It is thus immediate that our algorithm is in-place. For the complexity, \cref{eucldiv:diverase,eucldiv:sp}
  dominate the cost.  Using the exact complexity for \cref{eucldiv:diverase} given in \cref{cor:division-inplace-erasing}, one can
  deduce easily that Algorithm~\ref{alg:eucldiv} requires
  $\left(\lambda_m(\frac{c+1}{2}+\frac{1}{c})+ \lambda_s(2+\frac{1}{c})\right)\M(m)+ \bigO{m \log n}$ operations in $\K$.
\end{proof}

Using time-efficient products with $\lambda_m = \lambda_s=1$, $c_m=4$ and $c_s=3$ yields a complexity $\simeq 6.29 \M(m)$, which
is roughly $6.29/4 = 1.57$ times slower than the most time-efficient out-of-place algorithm.

\section{Multipoint evaluation and interpolation}

In this section, we present in-place algorithms for the two related problems of multipoint evaluation and interpolation. We first
review both classical algorithms and their space-efficient variants.

\subsection{Space complexity of classical algorithms}\label{ssec:spacecomp-evalinterp}

\paragraph{Multipoint evaluation}
Given $n$ elements $a_1$, \dots, $a_n$ of $\K$ and a size-$n$ polynomial $F\in\KX$, multipoint evaluation aims to compute
$F(a_1)$, \dots, $F(a_n)$. While the naive approach using Horner scheme leads to a quadratic complexity, the fast approach of
\cite{BorodinMoenck:1972:FastMultiPointEval} reaches a quasi-linear complexity $\bigO{\M(n)\log(n)}$ using a divide-and-conquer
approach and the fact that $F(a_i)=F \bmod (X-a_i)$. As proposed in~\cite{Bostan:2003} this complexity can be sharpened to
$(\lambda_m+\frac{1}{2}\lambda_f)\M(n)\log(n)+\bigO{\M(n)}$ using the transposition principle.

The fast algorithms are based on building the so-called \emph{subproduct tree}~\cite[Chapter~10]{Gathen:2013} whose leaves contain
the $(X-a_i)$'s and whose root contains the polynomial $\prod_{i=1}^n (X-a_i)$.  This tree contains $2^i$ degree-$n/2^i$ monic
polynomials at level $i$, and can be stored in exactly $n\log n$ registers if $n$ is a power of two. The fast algorithms then
require $n\log(n)+\bigO{n}$ registers as work space.
Here, because the space complexity constants \(c_f,c_m,c_s\) do not appear in the leading term \(n\log(n)\) of space usage, we can
always choose the fastest underlying multiplication routines, so the computational cost for this approach is simply
\(\tfrac{3}{2}\M(n)\log(n)+\bigO{\M(n)}\).

As remarked in~\cite{vonzurGathen1992}, one can easily derive a fast variant that uses only $\bigO{n}$ extra space.  In
particular, \cite[Lemma~2.1]{vonzurGathen1992} shows that the evaluation of a size-$n$ polynomial $F$ on $k$ points
$a_1$, \dots, $a_k$ with $k\leq n$ can be done at a cost $\bigO{\M(k)(\frac{n}{k}+\log(k))}$ with $\bigO{k}$ extra space.

We provide a tight analysis of this algorithm, starting with the \emph{balanced case} $k=n$, \emph{i.e.}\ the number of evaluation points
is equal to the size of $F$.  The idea of the algorithm is to group the points in \(\ceil{\log(n)}\) groups of
\(\floor{n/\log(n)}\) points each, and to use standard multipoint evaluation on each group, by first reducing $F$ modulo the root
of the corresponding subproduct tree.  The complexity analysis of this approach is given in the following lemma. Observe that here
too, the constants $\lambda_s, c_s$, etc., do not enter in since we can always use the fastest out-of-place subroutines without
affecting the $\bigO{n}$ term in the space usage.

\begin{restatable}{lem}{mpeevallinspacebalanced}\label{lem:mpeval-linspace-balanced}
  Given \(F\in\KX\) of size $n$ and \(a_1,\ldots,a_n\in\K\),
  one can compute \(F(a_1),\ldots,F(a_n)\) using
  \(\tfrac{7}{2}\M(n)\log(n) + \bigO{\M(n)}\) operations in $\K$ and
  \(n+\bigO{\frac{n}{\log(n)}}\) extra registers.
\end{restatable}

\begin{proof} \sloppy
  Computing each subproduct tree on \(\bigO{n/\log(n)}\) points can be done in time
  $\tfrac{1}{2}\M(n/\log(n))\log(n) \le \tfrac{1}{2}\M(n)$ and space \(n + \bigO{n/\log(n)}\).  The root of this tree is a
  polynomial of degree at most \(n/\log(n)\).  Each reduction of $F$ modulo such a polynomial takes time
  \(2\M(n) + \bigO{n/\log(n)}\) and space \(\bigO{n/\log(n)}\) using the balanced Euclidean division
  algorithm
  from \cref{ssec:classical-complexity}.  Each multi-point evaluation of the reduced polynomial on \(n/\log(n)\) points, using the
  pre-computed subproduct tree, takes $\M(n/\log(n))\log(n) + \bigO{\M(n/\log(n))}$ operations in $\K$ and \(\bigO{n/\log(n)}\)
  extra space~\cite{Bostan:2003}.

  All information except the evaluations from the last step --- which are written directly to the output space --- may be
  discarded before the next iteration begins. Therefore the total time and space complexity are as stated.
\end{proof}

When the number of evaluation points $k$ is large compared to the size $n$ of the polynomial $F$, we can simply repeat the
approach of \cref{lem:mpeval-linspace-balanced} \(\ceil{k/n}\) times.  The situation is more complicated when $k\le n$, because
the output space is smaller.  The idea is to compute the degree-$k$ polynomial $M$ at the root of the product tree, reduce $F$
modulo $M$ and perform balanced $k$-point evaluation of $F\bmod M$.

\begin{restatable}{lem}{mpevallinspace}\label{lem:mpeval-linspace}
  Given $F\in\KX$ of size $n$ and $a_1$, \dots, $a_k\in\K$, one can compute $F(a_1)$, \dots, $F(a_k)$ using
  \(2\lambda_s\M(n) + 4\M(k)\log(k) + \bigO{n+\M(k)\loglog(k)}\)
  operations in $\K$ and
  \((c_s+2)k + \bigO{k/\log(k)}\)
  extra registers.
\end{restatable}

\begin{proof}
  Computing the root $M$ of a product tree proceeds in two phases. For the bottom levels of the tree, we use the fastest
  out-of-place full multiplication algorithm that computes the product of two size-$t$ polynomials in time $\M(t)$ and space
  \(\bigO{t}\).  Then, only for the top \(\loglog(n)\) levels, do we switch to an in-place full product algorithm 
  from~\cite{GioGreRo19}, which has time \(\bigO{\M(t)}\) but only \(\bigO{1}\) extra space.  The result is that $M$ can be computed
  using \(\tfrac{1}{2}\M(k)\log(k)+\bigO{\M(k)\loglog(k)}\) operations in $\K$ and \(k+\bigO{k/\log(k)}\) registers.

  Then, we reduce $F$ modulo $M$.  By \cref{lem:remainder}, this is accomplished in time \(2\lambda_s\M(n) + \bigO{n + \M(k)}\)
  and space \((c_s+2)k\).  Adding the cost of the $k$-point evaluation of \cref{lem:mpeval-linspace-balanced} completes the proof.
\end{proof}

\paragraph{Interpolation}
Interpolation is the inverse operation of multipoint evaluation, that is, to reconstruct a size-$n$ polynomial $F$ from its
evaluations on $n$ distinct points $F(a_1)$, \dots, $F(a_n)$. The classic approach using Lagrange's interpolation formula has a
quadratic complexity~\cite[Chapter~5]{Gathen:2013} while the fast approach of~\cite{BorodinMoenck:1972:FastMultiPointEval} has
quasi-linear time complexity $\bigO{\M(n)\log(n)}$. We first briefly recall this fast algorithm.

Let $M(X)=\prod_{i=1}^n (X-a_i)$ and $M'$ its derivative. Noting that $\frac{M}{X-a_i}(a_i) = M'(a_i)$ for $1\le i\le n$, we have
\begin{equation}\label{eq:fastinterpolation}
  F(X)=M(X)\sum_{i=1}^n \frac{F(a_i)/M'(a_i)}{X-a_i}.
\end{equation}
Hence the fast algorithm of~\cite{BorodinMoenck:1972:FastMultiPointEval} consists in computing $M'(X)$ and its evaluation on each
$a_i$ through multipoint evaluation, and then to sum the $n$ fractions using a divide-and-conquer strategy. The numerator of the
result is then $F$ by~Equation~\eqref{eq:fastinterpolation}.

If the subproduct tree over the \(a_i\)'s is already computed, this gives all the denominators in the rational fraction sum.
Using the same subproduct tree for evaluating $M'$ and for the rational fraction sum gives the fastest interpolation algorithm,
combining the textbook method~\cite{Gathen:2013} with the multi-point evaluation of~\cite{Bostan:2003}.  The total computational
cost is only \(\tfrac{5}{2}\M(n)\log(n) + \bigO{\M(n)}\), while the space is dominated by the size of this subproduct tree,
\(n\log(n) + \bigO{n}\).

A more space-efficient approach can be derived using linear-space multipoint evaluation. Since the subproduct must be essentially
recomputed on the first and last steps, the total running time is \((2\lambda_f + \tfrac{7}{2})\M(n)\log(n) + \bigO{\M(n)}\),
using \((2 + \tfrac{1}{2}c_f)n + \bigO{n/\log(n)}\) registers.  This approach can be improved in two ways: first by again grouping
the interpolation points and re-using the smaller subproduct trees for each group, and secondly by using an in-place full
multiplication algorithm from~\cite{GioGreRo19} to combine the results of each group in the rational function summation. A
detailed description of the resulting algorithm, along with a proof of the following lemma, can be found in \cref{app:interp-linspace}.

\begin{restatable}{lem}{interplinspace}\label{lem:interp-linspace}
  Given \(a_1,\ldots,a_n\in\K\) and \(y_1,\ldots,y_n\in\K\),
  one can compute \(F\in\KX\) of size $n$ such that
  \(F(a_i)=y_i\) for \(1\le{}i\le{}n\) using
  \(5\M(n)\log(n) + \bigO{\M(n)\loglog(n)}\)
  operations in $\K$ and
  \(2n + \bigO{n/\log(n)}\)
  extra registers.
\end{restatable}

\subsection{In-place multipoint evaluation}
In order to derive an in-place algorithm we make repeated use of the unbalanced multi-point evaluation with linear space to
compute only $k$ evaluations of the polynomial $F$ among the $n$ original points. The strategy is to set $k$ as a fraction of $n$
to ensure that $n-k$ is large enough to serve as extra space. Applying this strategy on smaller and smaller values of $k$ leads to
\cref{alg:mpeval}, which is an in-place algorithm with the same asymptotic time complexity $\bigO{\M(n)\log(n)}$ as out-of-place
fast multipoint evaluation.

\begin{algorithm}
  \caption{In-Place Multipoint Evaluation (\textsc{InPlaceEval})}
\begin{algorithmic}[1]
  \Input $F\in\KX$ of size $n$ and $(a_1,\dots,a_n)\in\K^n$;
  \Output $R=(F(a_1),\dots, F(a_n))$
  \Required \textsc{Eval} of space complexity $\le(c_s+2)k$ as in \cref{lem:mpeval-linspace}
  \State $s\gets 0$,\quad $k\gets\floor{n/(c_s+3)}$
  \While{$k>0$}
    \State $R_{[s..s+k[} \gets \textsc{Eval}(F, a_s, \dotsc, a_{s+k})$
      \label{step:mpeval:eval}
      \Comment \WS $R_{[s+k..n[}$
    \State $s \gets s+k$
    \State $k\gets \floor{\frac{n-s}{c_s+3}}$
  \EndWhile
  \State $R_{[s..n[} \gets \textsc{Eval}(F,a_s,\dotsc,a_{n})$ \Comment constant space
\end{algorithmic}
\label{alg:mpeval}
\end{algorithm}

\begin{thm}\label{thm:inplace-multieval}
  Algorithm~\ref{alg:mpeval} is correct. It uses $\bigO{1}$ extra space and
  \(\left(4 + 2\lambda_s/\log(\frac{c_s+3}{c_s+2})\right)\!\M(n)\log(n) + \bigO{\M(n)\loglog n}\) operations in $\K$.
\end{thm}

\begin{proof}
  The correctness is obvious as soon as \textsc{Eval} is correct.  By the choice of $k$ and from the extra space bound of
  \textsc{Eval} from \cref{lem:mpeval-linspace}, \cref{step:mpeval:eval} has sufficient work space, and therefore the entire
  algorithm is in-place.
  The sequence \(k_i = \frac{(c_s+2)^{i-1}}{(c_s+3)^i}n\), for $i=1,2,\ldots$, gives the values of $k$ in each iteration.  Then
  \(\sum_i k_i \le n\) and the loop terminates after at most $\ell\log(n)$ iterations, where
  \(\ell \le 1/\log(\frac{c_s+3}{c_s+2})\).  Applying \cref{lem:mpeval-linspace}, the cost of the entire algorithm is therefore
  dominated by \(\sum_{1\le{}i\le\ell} \left(2\lambda_s\M(n) + 4\M(k_i)\log(k_i)\right)\), which is at most
  \((2\lambda_s\ell+4)\M(n)\log(n)\).
\end{proof}

Using a time-efficient short product with $\lambda_s=1$ and $c_s=3$ yields a complexity $\simeq 11.61 \M(n)\log n$, which is
roughly $11.61/1.5 = 7.74$ times slower than the most time-efficient out-of-place algorithm. 

\subsection{In-place interpolation}

Let \((a_1,y_1),\ldots,(a_n,y_n)\) be $n$ pairs of evaluations, with the $a_i$'s pairwise distinct. Our goal is to compute the
unique size-$n$ polynomial \(F\in\KX\) such that \(F(a_i)=y_i\) for \(1\le{}i\le{}n\), with an in-place algorithm.  Our first aim
is to provide a variant of polynomial interpolation that computes $F\bmod X^k$ using $\bigO{k}$ extra space.  Without loss of
generality, we assume that $k$ divides $n$. For $i=1$ to $n/k$, let $T_i= \prod_{j=1+k(i-1)}^{ki} (X-a_j)$ and $S_i=M/T_i$ where
$M = \prod_{i=1}^n (X-a_i)$. Note that $S_i = \prod_{j\neq i} T_j$.  One can rewrite \cref{eq:fastinterpolation} as
\begin{equation}\label{eq:fastinterpolation2}
  F(X) = M(X)\adjustlimits{\sum^{n/k}}_{i=1} {\sum^{ki}}_{j=1+k(i-1)} \frac{F(a_j)}{M'(a_j)}\frac{1}{(X-a_j)}
       = M(X)\sum_{i=1}^{n/k} \frac{N_i(X)}{T_i(X)} = \sum_{i=1}^{n/k}  N_i(X)S_i(X)
\end{equation}
for some size-$k$ polynomials $N_1$, \dots, $N_{n/k}$.  One may remark that the latter equality can also be viewed as an instance
of the chinese remainder theorem where $N_i=F/S_i \bmod T_i$ (see~\cite[Chapter~5]{Gathen:2013}).  To get the first $k$ terms of
the polynomial $F$, we only need to compute
\begin{equation}\label{eq:fmodxk}
  F \bmod X^k= \sum_{i=1}^{n/k} N_i (S_i \bmod X^k) \bmod X^k.
\end{equation}
One can observe that $M'(a_j)=(S_i \bmod T_i)(a_j)T_i'(a_j)$ for $k(i-1)<j\leq ki$. Therefore, \cref{eq:fastinterpolation2}
implies that $N_i$ is the unique size-$k$ polynomial satisfying $N_i(a_j)=(F/S_i \bmod T_i)(a_j)$ and can be computed using
interpolation.  One first computes $S_i \bmod T_i$, evaluates it at the $a_j$'s, performs $k$ divisions in $\K$ to get each
$N_i(a_j)$ and finally interpolates $N_i$.

\medskip
Our second aim is to generalize the previous approach when some initial coefficients of $F$ are known.  Writing
$F = G + X^s H$ where $G$ is known, we want to compute $H\bmod X^k$ from some evaluations of $F$.  Since $H$ has size at
most $(n-s)$, only $(n-s)$ evaluation points are needed. Therefore, using \cref{eq:fastinterpolation2} with
$M=\prod_{i=1}^{n-s}(X-a_i)$, we can write
  \begin{equation}\label{eq:fastinterpolation3}
    H(X) = M(X) \adjustlimits{\sum^{({n-s})/{k}}}_{i=1} {\sum^{ki}}_{j=1+k(i-1)} \frac{F(a_j) - G(a_j)}{a_j^sM'(a_j)}\frac{1}{(X-a_j)}.
  \end{equation}
  This implies that $H \bmod X^k$ can be computed using the same approach described above by replacing $F(a_j)$ with
  $H(a_j)=(F(a_j) -G(a_j))/a_j^s$. We shall remark that the $H(a_j)$'s can be computed using multipoint evaluation and fast
  exponentation.  Algorithm~\ref{alg:partinterpol} fully describes this approach.

\begin{algorithm}
\caption{Partial Interpolation (\textsc{PartInterpol})}
\begin{algorithmic}[1]
  \Input $G\in\KX$ of size $s$ and $(y_1,\dots,y_{n-s})$, $(a_1,\dots,a_{n-s})$ in $\K^{n-s}$; an integer $k\le n-s$
  \Output $H\bmod X^k$ where $F=G+X^{s}H\in\KX$ is the unique size-$n$ polynomial s.t. $F(a_i) = y_i$ for $1\leq i\leq n-s$
  \For{$i = 1$ to $(n-s)/k$}
    \State $S_i^k\gets 1$, $S_i^T\gets 1$ \label{partinterpol:init}
    \State $T_i\gets\prod_{j=1+k(i-1)}^{ki} (X-a_j)$    \Comment Fast divide-and-conquer
    \For{$j = 1$ to $(n-s)/k$, $j\neq i$}
        \State $T_j\gets\prod_{t=1+k(j-1)}^{kj} (X-a_t)$    \Comment Fast divide-and-conquer
        \State $S_i^k\gets S_i^k\times T_j\mod X^k$         \Comment $S_i^k = S_i\bmod X^k$
                                                            \label{partinterpol:Sik}
        \State $S_i^T\gets S_i^T\times T_j\mod T_i$         \Comment $S_i^T = S_i\bmod T_i$
                                                            \label{partinterpol:SiT}
    \EndFor
    \State $G^T \gets G \bmod T_i$ \label{partinterpol:mod}
    \State $(b_1,\dotsc,b_k)\gets\textsc{Eval}(S_i^T,a_{1+k(i-1)},\dotsc, a_{ki})$ \label{partinterpol:eval}
    \Statex \hspace{.485cm}$(z_1,\dotsc,z_k)\gets\textsc{Eval}(G^T,a_{1+k(i-1)},\dotsc, a_{ki})$
    \For{$j=1$ to $k$} 
        \State $b_j\gets (y_{j+k(i-1)}-z_j)/(a_{j+k(i-1)}^s b_j)$ 
    \EndFor
    \State $N_i\gets\textsc{Interpol}((z_1,\dotsc, z_k), (b_1,\dotsc,b_k))$ \label{partinterpol:interpol}
    \State $H_{[0..k[}\gets H_{[0..k[} + N_i S_i^k \bmod X^k$
  \EndFor 
\end{algorithmic}
\label{alg:partinterpol}
\end{algorithm}

\begin{lem}\label{lem:PartInterpolation} \sloppy
  Algorithm~\ref{alg:partinterpol} is correct. It requires $6k+\bigO{k/\log k}$ extra space and it uses
  $\left(\frac{1}{2}(\frac{n-s}{k})^2+ \frac{23}{2}\frac{n-s}{k}\right) \M(k)\log(k) + (n-s)\log(s)+
  \bigO{(\frac{n-s}{k})^2\M(k)\loglog k}$ operations in $\K$.
\end{lem}

\begin{proof}
  The correctness follows from the above discussion. In particular, note that the polynomials $S_i^k$ and $S_i^T$ at
  \cref{partinterpol:Sik,partinterpol:SiT} equal $S_i\bmod X^k$ and $S_i\bmod T_i$ respectively. Furthermore,
  $z_j=G(a_{j+k(i-1)})$ since $G(a_{j+k(i-1)})= (G \bmod T_i)(a_{j+k(i-1)})$.  Hence, \cref{partinterpol:interpol} correctly
  computes the polynomial $N_i$ and the result follows from \cref{eq:fmodxk,eq:fastinterpolation3}.

  From the discussion in Section~\ref{ssec:spacecomp-evalinterp}, we can compute each $T_i$ in
  $1/2\M(k)\log(k)+\bigO{M(k)\loglog k}$ operations in $\K$ and $k$ extra space. \Cref{partinterpol:eval} requires some care as we can share some
  computation among the two \emph{equal-size} evaluations. Indeed, the subproduct trees induced by this computation are identical
  and thus can be computed only once. Using Lemma~\ref{lem:mpeval-linspace-balanced}, this amounts to
  $\frac{13}{2}\M(k)\log (k)+\bigO{\M(k)}$ operations in $\K$ using $k+\bigO{k/\log k}$ extra space.  \Cref{partinterpol:interpol} can be done in
  $5\M(k)\log(k)+\bigO{M(k)\loglog k}$ operations in $\K$ and $2k+\bigO{k/\log k}$ extra space using
  Lemma~\ref{lem:interp-linspace}.  Taking into account the $n-s$ exponentations $a_j^s$, and that other steps have a complexity
  in $\bigO{\M(k)}$, the cost of the algorithm is
  \[
    \left(\frac{1}{2}\left(\frac{n-s}{k}\right)^2+ \frac{23}{2}\frac{n-s}{k}\right) \M(k)\log(k) + (n-s)\log(s)
    +O\left(\left(\frac{n-s}{k}\right)^2\M(k)\loglog k\right).
\]
  
  We show that $6k+\bigO{k/\log k}$ extra registers are enough to implement this algorithm.  At \cref{partinterpol:SiT}, the
  polynomials $T_i, T_j, S_i^k, S_i^T$ must be stored in memory. The computation involved at this step requires only $2k$ extra
  registers as $S_i^T \times T_j \bmod T_i$ can be computed with an in-place full product (stored in the extra registers) followed
  by an in-place division with remainder using the registers of $S_i^T$ and $T_j$ for the quotient and remainder storage.  Using
  the same technique \cref{partinterpol:mod} requires only $k$ extra space as for
  \crefrange{partinterpol:init}{partinterpol:Sik}. At \cref{partinterpol:eval}, we need $3k$ registers to store $G_T,S_i^T,S_i^k$
  and $2k$ registers to store $(b_1,\dotsc,b_k)$ and $(z_1,\dotsc,z_k)$, plus $k+\bigO{k/\log k}$ extra register for the
  computation. At \cref{partinterpol:interpol} we re-use the space of $G^T,S_i^T$ for $N_i$ and the extra space of the computation
  which implies the claim.
\end{proof}

We can now provide our in-place variant for fast interpolation.
\begin{algorithm}
\caption{In-Place Interpolation (\textsc{InPlaceInterpol})}
\begin{algorithmic}[1]
  \Input $(y_1,\dots,y_n)$ and $(a_1,\dots,a_n)$ of size $n$  such that $a_i,y_i\in\K$;
  \Output $F\in\KX$ of size $n$, such that $F(a_i) = y_i$ for $0\leq i\leq n$.
  \Required \textsc{PartInterpol} with space complexity $\le \cpi k$
  \State $s\gets 0$
  \While{$s < n$}
    \State $k \gets \floor{\frac{n-s}{\cpi+1}}$ 
    \If{$k = 0$} $k\gets n-s$     \EndIf
    \State $Y,A \gets (y_1,\dots,y_{n-s}),(a_1,\dots,a_{n-s})$
    \State $F_{[s..s+k[} \gets \textsc{PartInterpol}(F_{[0..s[}, Y, A, k)$ \label{interpol:part}
    \State $s \gets s+k$

  \EndWhile 
\end{algorithmic}
\label{alg:interpol}
\end{algorithm}

\begin{thm}\label{thm:interpol}
  Algorithm~\ref{alg:interpol} is correct. It uses $\bigO{1}$ extra space and at most
  $\frac{1}{2}(c^2+23c)\M(n)\log n+\bigO{\M(n)\loglog n}$ operations in $\K$, where $c=1+\cpi$.
\end{thm}

\begin{proof}
  The correctness is clear from the correctness of Algorithm \textsc{PartInterpol}.  To ensure that the algorithm uses $\bigO{1}$
  extra space we notice that at \cref{interpol:part}, $F_{[s+k..n[}$ can be used as work space.  Therefore, as soon as
  $\cpi k\le n-s-k$, that is, $k\le \frac{n-s}{\cpi+1}$, this free space is enough to run \textsc{PartInterpol}.  Note that when
  $k=0$, $n-s < c_{pi}+1$ is a constant, which means that the final computation can be done with $\bigO{1}$ extra space.  Let
  $k_1$, $k_2$, \dots, $k_t$ and $s_1$, $s_2$, \dots, $s_t$ be the values of $k$ and $s$ taken during the course of the algorithm.
  Since $s_i=\sum_{j=1}^i k_j \leq n$ with $s_0=0$, we have $k_i\le\lambda n(1-\lambda)^{i-1}$, and
  $s_i\ge n(1 - (1-\lambda)^{i})$ where $\lambda=\frac{1}{\cpi+1}$.  The time complexity $T(n)$ of the algorithm satisfies
  \[
    T(n)\le \sum_{i=1}^t \left(\frac{c^2}{2}+\frac{23c}{2}\right) \M(k_i)\log(k_i) + \sum_{i=1}^t(n-s_{i-1})\log(s_{i-1})
    + \bigO{c^2\M(k_i)\loglog k_i}
  \]
  since $\frac{n-s_{i-1}}{k_i}\le c=\cpi+1$ by definition of $k_i$. Moreover, we have
  \(\sum_{i=1}^t\M(k_i)\log(k_i)\le \M(\sum_i k_i)\log n\le \M(n)\log(n)\).  By definition of $s_i$, we have
  $n-s_i\le n(1-\lambda)^i$ which gives
  \[
    \sum_{i=1}^t(n-s_{i-1})\log(s_{i-1}) \leq  n \log(n) \sum_{i=1}^t (1-\lambda)^i \le (\cpi+1) n \log n.
  \]
  This concludes the proof.
\end{proof}

Since $\cpi < 6+\epsilon$ for any $\epsilon>0$, the complexity can be approximated to $105 \M(n)\log(n)$, which is $42$ times slower
than the fastest interpolation algorithm (see Table~\ref{table:comp}).

\subsection*{Acknowledgments}
We thank Grégoire Lecerf, Alin Bostan and Michael Monagan for pointing out the references~\cite{vonzurGathen1992,Monagan93}.

\newcommand{\Gathen}{\relax}\newcommand{\Hoeven}{\relax}

\clearpage

\appendix

\section{Interpolation with linear space}
\label{app:interp-linspace}

The algorithm proceeds as:
\begin{enumerate} \def\labelenumi{(\arabic{enumi})}
  \item Run the subproduct tree algorithm for each group of
    \(n/\log(n)\) interpolation points, saving only the roots of each
    subtree \(M_1,\ldots,M_{\ceil{\log(n)}}\), using fast out-of-place
    full multiplications.
  \item Run the subproduct tree algorithm over these \(M_i\)'s to
    compute the root \(M\), using in-place full multiplications 
    from~\cite{GioGreRo19}, discarding other nodes in the tree.
  \item Compute the derivative \(M'\) in place.
  \item Compute the remainders \(M'\bmod M_i\) for
    \(1\le{}i\le{}\ceil{\log(n)}\), using the balanced (with
    precomputation) algorithm described in
    \cref{ssec:classical-complexity}.
    The size-$n$ polynomial $M'$ may now be discarded.
  \item For each group $i$, compute the full subproduct tree over its
    $n/\log(n)$ points. Use this to perform multi-point evaluation
    of \(M'\bmod{}M_i\) over the \(n/\log(n)\) points of that group
    only, and then compute the partial sum of
    \eqref{eq:fastinterpolation} for that group's points.
    Discard the subproduct tree but save the rational function partial
    sum for each group.
  \item Combine the rational functions for the \(\ceil{\log(n)}\) groups
    using a divide-and-conquer strategy, employing again the in-place full
    multiplications from~\cite{GioGreRo19}.
\end{enumerate}

The following lemma gives the complexity of this linear-space
interpolation algorithm.

\interplinspace*

\begin{proof}
  Steps (1) and (5) collectively involve, for each group,
  two subproduct tree computations,
  one multi-point evaluation, and one rational function summation over
  each group, for a total of \(3\M(n)\log(n) + \bigO{\M(n)}\) time.
  Step (4) contributes another
  \(2\M(n)\log(n) + \bigO{\M(n)}\) operations in $\K$.
  In steps (2) and (6), the expensive in-place multiplications are used
  only for the top \(\ceil{\loglog(n)}\) levels of the entire subproduct
  tree, so this contributes only \(\bigO{\M(n)\loglog(n)}\).

  For the space, note that the size-$n$ output space may be used during
  all steps until the last to store intermediate results.
\end{proof}

\end{document}